\documentclass[sigplan]{acmart}

\usepackage{preamble}

\newcommand{\nc}[1]{{\color{cyan}\grumbler{Natacha}{#1}}}
\newcommand{\agc}[1]{{\color{blue}\grumbler{Allen}{#1}}}

\newcommand{\neil}[1]{{\color{red}\grumbler{Neil}{#1}}}

\title{It’s not a lie if you don’t get caught: \\  simplifying
reconfiguration in SMR through dirty logs}
\date{}
\author{Allen Clement}
\affiliation{
  \institution{Subzero Labs}
  \country{}
}
\author{Natacha Crooks}
\affiliation{
  \institution{Subzero Labs}
  \country{}
}
\affiliation{
  \institution{UC Berkeley}
  \country{}
}
\author{Neil Giridharan}
\affiliation{
  \institution{Subzero Labs}
    \country{}
}
\affiliation{
  \institution{UC Berkeley}
  \country{}
}

\author{Alex Shamis}
\affiliation{
  \institution{Subzero Labs}
  \country{}
}

\begin{document}

\begin{abstract}

Production state-machine replication (SMR) implementations are complex, multi-layered architectures comprising data dissemination, ordering, execution, and reconfiguration components. Existing research consensus protocols rarely discuss reconfiguration. Those that do tightly couple membership changes to a specific algorithm. This prevents the independent upgrade of individual building blocks and forces expensive downtime when transitioning to new protocol implementations. Instead, modularity is essential for maintainability and system evolution in production deployments. We present \sys{}, a reconfiguration engine designed to treat consensus protocols as interchangeable modules. By introducing a distinction between a consensus protocol's \textit{inner log} and a sanitized \textit{outer log} exposed to the RSM node, \sys{} allows engineers to upgrade membership, failure thresholds, and the consensus protocol itself independently and with minimal global downtime. Our initial evaluation on the Rialo blockchain shows that this separation of concerns enables a seamless evolution of the SMR stack across a sequence of diverse protocol implementations.
\end{abstract}

\maketitle
\section{Introduction}

State-machine replication (SMR)~\cite{schneider1990smr} has emerged as the fundamental fault-tolerant core of modern distributed
applications. Most database systems use crash-fault-tolerant SMR protocols such as Raft or Paxos to defend against data loss in the event of machine failures. 
Spanner, Google's geo-distributed flagship database uses Multi-Paxos at its core; Neo4J, a popular graph database, uses Raft, while Cassandra leverages Accord~\cite{antoniadis2023accord},
a leaderless consensus protocol. Blockchain systems such as Rialo~\cite{rialo2026blockchain} or Sui~\cite{sui2026blockchain} similarly derive their security guarantees from Byzantine-fault-tolerant (BFT) consensus protocols like PBFT~\cite{castro1999pbft}, 
Hotstuff~\cite{yin2019hotstuff}, Mysticeti~\cite{babel2025mysticeti} or Autobahn~\cite{giridharan2024autobahn}.

Consensus protocols have a hard job: they must offer strong safety and liveness guarantees at high throughput and low latency. Research in this space is 
rich~\cite{giridharan2024autobahn,castro1999pbft,babel2025mysticeti,cowling2006hq, kotla2007zyzzyva, clement2009aardvark, miller2016honeybadgerbft,gueta2019sbft, neiheiser2021kauri, spiegelman2022bullshark, danezis2022tusk, antunes2024aleabft}. Consensus \textit{engineers}, however, have an even harder job! Not only must they deliver on the promises made by the protocol, 
but they must additionally face the operational realities of modern distributed systems, both technical and socio-technical. 

On the technical front, engineers must grapple with the realities of hardware degradation, network partitions and shifting load patterns. 
Whether it is a global database like Spanner migrating a shard between continents to follow user activity, permissioned blockchains like Sui or Rialo rotating 
its validator set to ensure forward security, or a graph database removing a misconfigured machine suffering from a fail-slow fault,  the system must support 
the dynamic addition and removal of consensus participants.  Unfortunately,  academic research on consensus primarily assumes a fixed universe of participants, $P$, 
known a priori and immutable for the duration of the system’s execution. In this idealised static model, the quorum size $Q$ is constant, the failure threshold $f$ is fixed, and 
the identity of every participant is hard-coded into the protocol state. A survey of recent consensus papers published at top security/systems venues show that only a small fraction had a detailed \textit{reconfiguration protocol} to change the participant set. 

On the socio-technical front, engineers must also navigate organisational changes, team restructurings, and evolving business requirements that impact system design and operation.
 There are two primary software engineering challenges: 1) modularity is essential for future-proofing 2) code evolves.  Modularity, even inside of the consensus module, 
 is key for maintainability. A realistic SMR implementation 
contains several components: the data dissemination layer, the ordering layer, the execution engine, and the reconfiguration logic among others. Each of these components should be modifiable and
 testable independently.  Moreover, these building-blocks
\textit{evolve} in ways that are often unpredictable. A database or blockchain company rarely deploys a single consensus protocol. Instead, it deploys a sequence of consensus implementations.
 Each implementation improves over time as the company matures. Sui, for instance, a popular blockchain system, recently announced a "v2" of their internal consensus algorithm Mysticeti~\cite{sui2025mysticetiv2} while Cassandra recently migrated to the leaderless Accord CFT consensus protocol~\cite{apache2023accord}. Overall correctness must not only hold across a dynamic set of participants,
  but across a sequence of consensus protocol implementations! 
Unfortunately,  all existing reconfiguration algorithms today tightly couple the reconfiguration logic with a \textit{single} consensus protocol. 
There does not currently exist a clean way to update consensus. Incremental updates are hard enough; dramatic changes, including adoption of a new protocol entirely, are expensive and take years. Evolution usually results in protocol downtime and code bases often lose the ability to access old committed transactions when consensus changes. 

In this paper, we argue that a practical reconfiguration protocol should achieve the following three properties: 1) arbitrary membership changes 2) full modularity 3) minimal downtime: 
\begin{itemize}[leftmargin=*]
\item \textbf{Arbitrary Membership Change} It should be possible to replace any and all participants from one configuration to the next. This consideration is especially important in blockchain systems where validators have stake associated with their votes.
\item  \textbf{Full Modularity} The reconfiguration logic should be strictly independent of all other components in the SMR node. It should be possible to transition the underlying consensus engine from one arbitrary consensus protocol to another, without knowing anything about the details, except for the fact that it satisfies the traditional safety and liveness definitions. 
\item  \textbf{Minimal Downtime} We contend that the operational benefits of satisfying Goals 1 and 2 outweigh the benefits of achieving optimal performance during a reconfiguration event. Reconfiguration is traditionally rare. As such, we can tolerate brief latency spikes during the reconfiguration process as long as downtime remains minimal. 
\end{itemize} 

This paper presents the design and implementation of a novel reconfiguration engine, \sys{}, that satisfies all three properties. The reconfiguration logic makes no assumption about the underlying consensus protocol used to order transactions, only that it satisfies the safety and liveness property of consensus.  Yet it achieves full generality: two epochs $i$ and $i+1$ can differ in their consensus protocol, membership, as well as failure thresholds. \sys{} achieves this with minimal downtime. 

To achieve this, we take as our starting point the ideas of Horizontal Paxos~\cite{lamport2010reconfiguring}, a classic reconfiguration protocol for crash-fault-tolerant (CFT) systems. Each consensus's log within an epoch provides a global consistent order across replicas that can be leveraged to disseminate reconfiguration information. To achieve full generality, however, we make a second observation: the log of operations generated by the consensus protocol itself need not necessarily be the log exposed by the consensus engine to other components of the SMR node. We can lie! Specifically, a transaction marked as committed in a consensus protocol's \textit{inner log} need not necessarily be marked committed in the \textit{outer log} exposed to other components, as long as the overall SMR properties are preserved.  Fundamentally, this idea is not new. Transactions submitted to multiple leaders must be deduplicated in most leader-based systems, BFT or CFT. Blockchain systems check external 
validity predicates prior to executing the transaction. 

Applied to reconfiguration, however, distinguishing the inner log and outer log in this way is powerful. It allows seamless transition between consensus epochs without making \textit{any} assumptions about the protocols themselves. For instance, it achieves even greater generality than the original Horizontal Paxos algorithm: \sys{} need not bound the number of concurrent/pending transactions in the system.

We implement this idea as part of the Rialo blockchain, cleanly modularising each RSM component for future upgradeability.

\section{System Model}
\label{sec:model}

Practical deployments of state machine replication protocols must support 1) a changing set of participants 2) updates to the consensus protocol itself. 
Existing formalisms rarely consider this reality. We first introduce formalisms that capture the long-term, evolving nature of real SMR deployments that 
must continue to provide safety and liveness guarantees across a sequence of changing configurations and consensus protocols.

An SMR system must satisfy the following top-level properties: \neil{We should probably put the propose/decide definitions and the part about how consensus 
outputs a log before here instead of after so these properties are easier to follow.}
\begin{itemize}
    \item \textbf{Safety.} If two honest nodes commit transactions $(j,x)$ and $(j,x')$, respectively, for the same log position $j$, then $x=x'$.
    \item \textbf{Liveness.} If an honest node proposes a transaction $x$, then all honest nodes eventually commit $(j,x)$ for some log position $j$.
    \item \textbf{Integrity.} A transaction $x$ appears in the log at most once. In particular, there do not exist two distinct log positions $j\neq j'$
     such that honest nodes decide $(j,x)$ and $(j',x)$.
    \item \textbf{External Validity.} For any committed transaction $x$, \textsc{ExVal}$(x)=true$, where \textsc{ExVal} is a predicate that checks whether 
    $x$ upholds all application invariants.
\end{itemize}

These are the properties that external clients observe. The SMR system can ingest this log to materialize shared state. Internally, however, the SMR system 
is composed of a sequence of epochs $e_1, e_2, \ldots$. Each epoch $e_i$ is associated with a membership configuration $M_i$, which specifies a set of $n_i$ 
participating nodes, and a consensus protocol $C_i$. The protocol $C_i$ exposes \textit{propose}$(x)$ and \textit{decide}$(j,x)$ events,
 where \textit{decide}$(j,x)$ indicates that transaction $x$ is chosen for log position $j$. We assume that $C_i$ satisfies the safety and 
 liveness properties for consensus - it outputs a totally ordered log, but make no additional assumptions about the internal structure of $C_i$.
Nodes in a configuration communicate over an asynchronous network that places no bounds on message delivery times. We assume that in an 
epoch $e_i$ with $n_i$ nodes, at most $f_i$ nodes may fail during execution. Failures follow either the crash fault model or the Byzantine fault model. 
In the crash fault model, a faulty node may halt at any time but otherwise behaves correctly. In the Byzantine fault model, faulty nodes may behave 
arbitrarily. To tolerate $f_i$ failures, we assume $n_i\geq 2f_i+1$ under crash faults and $n\geq 3f_i+1$ under Byzantine faults.

The system combines all of the inner logs to produce a final, sanitized, totally ordered log called the \textit{outer log} that satisfies the aforementioned SMR properties. This layered structure captures the reality of modern SMR systems, where the consensus protocol is but one component in a 
larger architecture.

\section{Design}
\label{sec:design}
In this section, we present the design of \sys{}, our modular architecture that supports the independent evolution of individual system components.

\begin{figure}[t!]
    \centering
    \includegraphics[width=\linewidth]{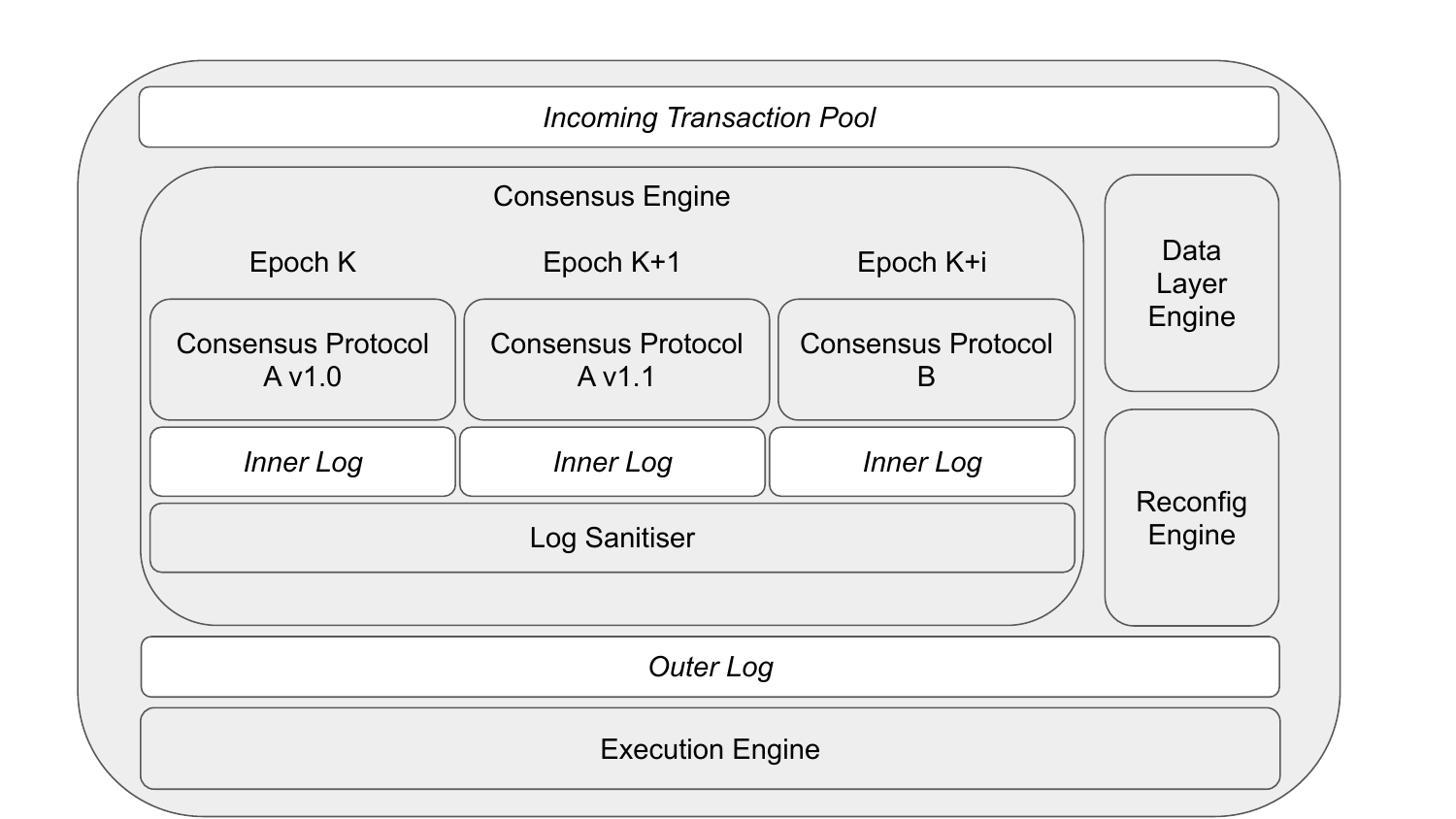}
    \caption{RSM Node SubComponents}
    \label{fig:validator_components}

\end{figure}

\subsection{RSM Node Architecture}
\label{subsec:rsm_architecture}
An RSM node typically consists of four key components (Figure~\ref{fig:validator_components}):  1) a data dissemination layer, 
tasked with disseminating 
transaction data to all nodes in the system 2) a consensus engine, tasked with outputting an ordered log of transactions 
consistently across 
all RSM replicas 3) an execution engine, tasked with ingesting this log to materialize the final application state 
4) the reconfiguration engine, 
tasked with upgrading the different components of the system. These components work closely together to maintain the 
correctness of the replicated 
state machine. In a production implementation of an RSM node, these units must be designed in a modular fashion such that they can be independently 
modified and upgraded over time.  This requires careful interface design, and clear separation of concerns at each component boundary. In this paper, 
we focus on the design of the reconfiguration engine, responsible for coordinating the transition between epochs. 

\subsection{Reconfiguration Engine}
\label{subsec:reconfig_engine}
In line with our stated objective, our reconfiguration engine, \sys{},  makes no assumption about the underlying consensus protocol used to order transactions, 
only that it satisfies the safety and liveness property of consensus (Section \ref{sec:model}).  Yet it achieves full generality: two epochs $i$ and $i+1$ 
can differ in their consensus protocol, membership, as well as failure thresholds. \sys{} achieves this with minimal downtime, by leveraging two observations:
1) a given consensus protocol within an epoch already outputs a totally ordered \textit{inner log} of transactions across all correct nodes in the epoch 2) 
this \textit{inner log} need not be directly executed, instead, it can be sanitised and transformed to produce a new \textit{outer} log. Only the outer log is exposed to the other node subcomponents.

The reconfiguration engine runs a three stage protocol ensuring a safe (and live) transition from one epoch configuration to the next. 
The \textit{prepare phase} ensures that members of the new configuration prepare their local state and components to begin processing 
transactions. The \textit{handover phase} coordinates between the two epoch configurations to seamlessly transfer control and establish a
 trust chain between each configuration. This trust chain provides evidence to external users that the previous configuration agrees and supports
  the next configuration. Finally, the \textit{shutdown phase} ensures that the outgoing members of the old configuration can safely wind down while
  preserving liveness.  The log sanitizer component at each replica regulates the translation from the consensus inner logs to the outer log exposed 
  to other components of the RSM node.

We next describe each phase in turn, illustrating the transition between epochs $i$ and $i+1$.  Epoch $i$ uses consensus protocol $C_i$ with membership 
$M_i$ and failure threshold $f_i$, while epoch $i+1$ uses consensus protocol $C_{i+1}$ with membership $M_{i+1}$ and 
failure threshold $f_{i+1}$.  Consensus protocol $C_{i}$ produces inner log $L_{i}$, while $C_{i+1}$ 
produces inner log $L_{i+1}$. By definition of consensus (Section~\ref{sec:model}), all correct nodes in epoch $i$ agree on the same log $L_{i}$, and all
 correct nodes in epoch $i+1$ agree on the same log $L_{i+1}$. Similarly, if a correct node sees $T_i$ at position $i$ in the log, all correct nodes 
 eventually see $T_i$ at position $i$ in the log. \neil{I'm not sure if this is true with our reconfiguration protocol, since an honest node in the old 
 configuration might commit a slot past the snapshot, which is not guaranteed to be committed by all honest} \agc{The thing that pops into my mind reading this is that 
 the core challenge of epoch change becomes safely appending $L_{i+1}$ onto an appropriate pre-fix of $L_{i}$}
For simplicity, we assign different identities to all nodes in $M_i$ and $M_{i+1}$. We write $R_{(i,k)}$ to denote the $k$-th replica in epoch $i$. In
 practice, there may be significant overlap between configurations. As shown in Figure~\ref{fig:validator_components}, nodes then simply run concurrent 
 epoch instances.

\par \textbf{Prepare Phase} An operator submits an \EpochChange{} transaction to consensus instance $C_i$. 
The logic by which the \EpochChange{} transaction is generated and submitted is outside the scope of this paper. It may be generated by a trusted operator, 
via an on-chain governance mechanism, or as a result of a protocol update. This transaction is treated as any other transaction in the RSM node. 
By properties of consensus, it will eventually appear in the log of consensus $C_i$ at some position $K$. The epoch change transaction consists of the identity of the current epoch being transitioned from $i$, the new epoch being transitioned to $e_{i+1}$. The latter additionally contains the new configuration parameters: the membership $M_{i+1}$, failure threshold $f_{i+1}$, and consensus protocol $C_{i+1}$.

\protocol{$R_{(i,k)}$: Detects \EpochChange{$i$, $e_{i+1}$} transaction} Upon commitment of the \EpochChange{} transaction at position $K$ in $L_i$, 
the log sanitizer of $R_{(i,k)}$ begins waiting for \Ready{} messages from all members of epoch $i+1$.  No other changes occur. Notably, the log sanitizer 
continues to process and output transactions from $L_i$ to the outer log $O$ for execution as normal. \nc{Is this too vague?}

\protocol{$R_{(i+1,k)}$: Detects \EpochChange{$i$, $e_{i+1}$} transaction} Upon observing an \EpochChange{} message, each incoming replica member $R_{(i+1,k)}$ of
 the new epoch $i+1$ synchronises its local components to begin processing transactions from epoch $i+1$. This includes initializing its consensus 
 protocol $C_{i+1}$, ensuring that it is sufficiently up-to-date with the current outer log, as well as synchronising application and execution state.   Once ready, it submits a
  \Ready{$i$, $i+1$, hash, $sig_k$} transaction to $C_{i}$ where $i$ and $i+1$ denote respectively the epochs being transitioned from and to, $hash$ is the hash of the \EpochChange{} transaction, and $sig_k$ is a signature of the \Ready{} message. By properties of consensus, this transaction will eventually appear in the log of consensus
   $C_i$ at some position $j$. \\

\protocol{$R_{(i+1,k)}$: Submits \Ready{} message as a transaction to $C_{i}$}

\par \textbf{Handover Phase} Two possible cases arise: 1) all \Ready{} transactions for epoch $i+1$ appear committed in the log before a new epoch change 
transaction is observed 2) a new epoch change transaction commits, preempting the initial epoch change. We discuss each scenario in turn. 
Crucially, every replica in epoch $i$ can deterministically decide which case has occurred by examining the log $L_i$. By properties of consensus, 
all correct replicas in epoch $i$ agree on the same log $L_i$.

\textit{Case 1: Successful Handover} If all \Ready{} transactions from members of epoch $i+1$ appear committed in the log before a new epoch change 
transaction is observed, $R_{i,k}$ considers the handover to be successful. It forms a \textit{handover certificate} for epoch $i$, indicating that 
epoch $i$ finished at position $h$ in log $L_i$, where $h$ is the log position of the last \Ready{} transaction. The handover certificate contains the id of the old epoch $i$, information about the new epoch $E_i$, the position of the handover certificate in the log $h$, as well as the hash of the previous handover certificate for epoch $i$ ($hash_i$).
 It signs the handover certificate and submits this as a \Done{} transaction to consensus $C_{i+1}$. Replica $R_{(i,k)}$  continues participating in protocol $C_i$ as normal. \\ 

\protocol{$R_{(i,k)}$: Forms signed \Handover{$i$, $E_{i+1}$, $h$, $hash_i$} certificate.  Submits as transaction to $C_{i+1}$}
Because the reconfiguration protocol makes no assumption about the underlying consensus protocol, it is possible that transactions 
continue to be committed in $L_i$ after position $h$. This is true for multi-proposer systems~\cite{stathakopoulou2019mirbft, giridharan2024autobahn,babel2025mysticeti}, as well as systems that allow for parallel consensus instances~\cite{castro1999pbft}. Therein lies the power of combining horizontal reconfiguration with the inner/outer log distinction.
 By definition, all replicas output the same inner log. They consequently all form the same handover certificate, and stop processing the inner log of epoch $i$ at the same position $h$ in $L_i$. 
 The log sanitizer of $R_{(i,k)}$ simply ignores all transactions at positions $h'>h$ in $L_i$
  and does not submit them to the outer log $O$. They will never appear committed to the overall SMR node. Submitting \Ready{} messages to
   consensus increases latency for new configurations to become active, but is key to achieving minimal downtime. Transactions continue to be
   committed and executed in epoch $i$ until the handover certificate is formed at position $h$. Only transactions after position $h$ are discarded.

\textit{Case 2: Preemption} Members of the new epoch may fail to reply, or fail to synchronize their state in a timely fashion such that they are ready
 to proceed in the new epoch. The protocol simply allows later epoch changes to preempt earlier ones. If a new \EpochChange{} transaction appears in 
 $L_i$ before all \Ready{} transactions from epoch $i+1$ have been committed, replica $R_{(i,k)}$ in epoch $i$ considers the handover to have been preempted.
  It informs the log sanitizer to stop listening for \Ready{} transactions from epoch $i+1$, and to instead listen for \Ready{} transactions from members
   of the new epoch $i+2$. The protocol then proceeds as before.

\par \textbf{Shutdown Phase} The shutdown phase ensures that members of epoch $i$ can safely wind down and terminate the epoch without violating safety or liveness. To achieve this,
 the protocol should make sure that 1) the new epoch is fully operational before old epoch members stop participating in the consensus protocol $C_i$ 2) 
 there is a single such epoch.   \\ 
\neil{General comment: the distinction between replicas in epoch i vs i+1 is a bit small and easy to miss with just the subscript, can we add a heading above?}
\agc{Would it be useful to use $incoming := replicas in epoch_{i+1}$ and $outgoing := replicas in epoch{i}$ ?}

\protocol{$R_{(i,k)}$: Detects $f_i + 1$ \Done{i, $e_{i+1}$, $hash_i$} transactions in $C_{i+1}$}
\vspace{2pt}
Upon detecting $f_i+1$ \Done{} transactions from distinct members of epoch $i$ in consensus $C_{i+1}$, replica $R_{(i,k)}$ 
considers epoch $i+1$ to be fully active. At least one correct member of epoch $i$ has submitted a \Done{} transaction
 (there is a maximum of $f_i$ malicious replicas), indicating that it has safely formed a handover certificate and submitted 
 it to consensus $C_{i+1}$. The presence of an honest replica's \Done{} transaction offers two guarantees: 1) 
  all subsequent members can recover sufficient state (as the \Done{} transaction includes a signed checkpoint and 2) the guarantee 
  that no other epoch is active. Honest nodes would never send conflicting \Done{} transactions. It can now safely stop participating in consensus protocol $C_i$ and shutdown the epoch. \\
\protocol{$R_{(i+1,k)}$: Detects $f_i + 1$ \Done{i, $e_{i+1}$, $hash_i$} transactions in $C_{i+1}$}
Upon detecting $f_i+1$ \Done{} transactions from members of epoch $i$ in consensus $C_{i+1}$, replica $R_{(i+1,k)}$
considers epoch $i+1$ to be fully active. At least one correct member of epoch $i$ has submitted a \Done{} transaction
 (there is a maximum of $f_i$ malicious replicas), indicating that it has safely formed a handover certificate and submitted it to consensus $C_{i+1}$.
By properties of consensus, all correct members of epoch $i+1$ will see these $f_i+1$ \Done{} transactions at the same positions in $L_{i+1}$.
The log sanitizer can now begin outputting committed transactions from $L_{i+1}$ to the outer log $O$ for execution.
 All other replicas in epoch $i+1$ proceed similarly.

\subsection{Worked Example}
We provide a short worked example to further explain our approach. Consider epochs $k$ and $k+1$ as shown in Figure~\ref{fig:epoch_transition}.
 Epoch $k$ uses consensus protocol $C_k$ (Mysticeti) with membership $M_k = \{R_{(k,1)}, R_{(k,2)}, R_{(k,3), R_{(k,4)}}\}$ and failure threshold
  $f_k = 1$, while epoch $k+1$ uses consensus protocol $C_{k+1}$ (PBFT) with membership $M_{k+1} = \{R_{(k+1,5)}, R_{(k+1,6)}, R_{(k+1,7), R_{(k+1,8)}}\}$ 
  and failure threshold $f_{k+1} = 1$.  \agc{in figure, i blieve the DONE should show up in the outer log?  or is it contained 
  in the epoch engine based on this descirption?}.

Transactions $T_1, T_2$ commit in inner log $L_k$ at positions $1$ and $2$. The log sanitizer forwards them to the outer log $O$.
 At position $3$, an epoch change transaction is initiated, transitioning the system from epoch $k$ to $k+1$ with a disjoint set of validators. The log sanitizer begins waiting to see \Ready{} transactions appear in $L_k$. While waiting for these transactions, $T_3$ and $T_4$ commit in $L_k$ and the log sanitizer forwards them to $O$. Upon detecting the final \Ready{} transaction at position $9$ in the log, the node forms a handover certificate finalising the epoch. It submits this certificate to $C_{i+1}$ in epoch $k+1$ and discards transactions $T_5$ and $T_6$ (Figure~\ref{fig:epoch_transition}). Upon seeing two \Done{} messages, epoch $k+1$ becomes active. $T_5$ is resubmitted by a client and commits at position 3 of inner log $L_{k+1}$. Transactions $T_7$ and $T_8$ commit at positions 4 and 5 of inner log $L_{k+1}$. All are forwarded to the outer log for execution.  The outer log exposes a single totally ordered log consisting of transactions $T_1, T_2, T_3, T_4, T_5, T_7, T_8$ and a trust chain of epoch transitions.
 
 \begin{figure}[h]
    \centering
    \includegraphics[width=\linewidth]{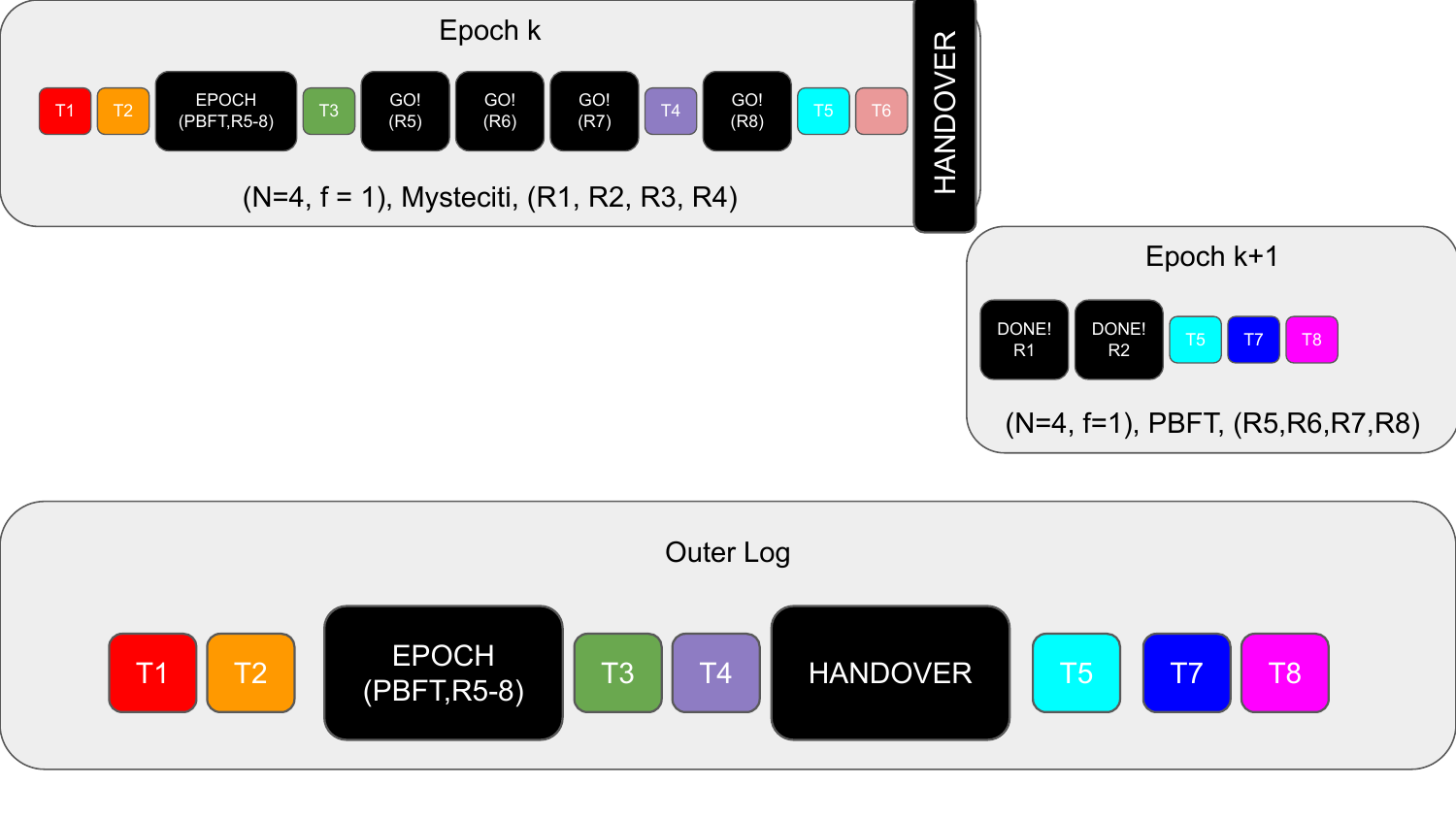}
    \caption{Epoch Transition Example}
    \label{fig:epoch_transition}    
\end{figure}

\section{Proofs}

We provide short proof sketches describing how our reconfiguration engine \sys{} preserves
the SMR properties of safety and liveness during epoch transitions for 
arbitrary consensus protocols and membership configurations.

\nc{Add the various theorems here if there is time}
\agc{properties that I find/found important:  
(1) Trust chain is unbroken, $i$ signs $HC_{i+1}$ 
(2) in epoch $i$, $HC_{i -> i+1}$ is formed only if $HC.epoch_change$ is the maximal preceding epoch change in $log{i}$
}

\begin{theorem}[Safety]
If two honest nodes decide $(j, x)$ and $(j, x')$ for the same outer log position $j$, then $x = x'$.
\end{theorem}

\begin{proof}
Suppose, for contradiction, that $x \neq x'$. Let $N_1$ and $N_2$ be two honest nodes such that $N_1$ decides $(j, x)$ and $N_2$ decides $(j, x')$. Each outer log position corresponds to an inner log position in some epoch. Let $k$ denote the inner log position corresponding to outer position $j$, and $e$ the corresponding epoch. For each epoch $i$, let $L_i$ and $L_i'$ denote the inner logs produced by epoch $i$ as observed by $N_1$ and $N_2$, respectively. Let $L_e^{\le k}$ (resp., $(L_e')^{\le k}$) denote the prefix of $L_e$ (resp., $L_e'$) truncated to its first $k$ entries.

Node $N_1$ computes $x$ by applying the deterministic log sanitizer $S$ to the concatenation
$L_1||L_2||\cdots|| L_{e-1}|| L_e^{\le k}$,
and similarly, $N_2$ computes $x'$ by applying $S$ to
$L_1'||L_2'||\cdots||L_{e-1}'||(L_e')^{\le k}$.

By the safety property of each consensus protocol $C_i$, all honest nodes agree on the contents of the inner log for epoch $i$. Therefore, for all $i<e$, we have $L_i = L_i'$, and moreover $L_e^{\le k} = (L_e')^{\le k}$. Since $S$ is deterministic, it follows that $x = x'$, contradicting our assumption.

Hence, $x = x'$, completing the proof.
\end{proof}

\section{Implementation and Evaluation}

\begin{figure*}[!htb]
    \centering
    \begin{subfigure}[t]{0.48\textwidth}
        \centering
        \includegraphics[width=\textwidth]{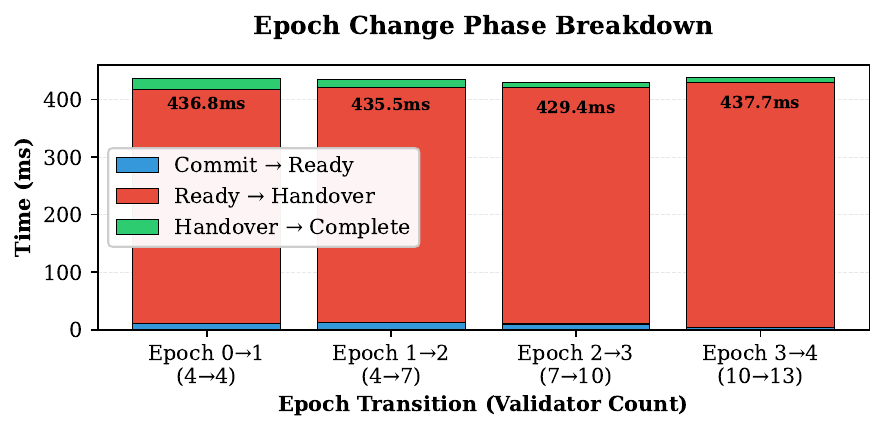}
        \caption{Breakdown of each transition showing time spent in each stage.}
        \label{fig:epoch-breakdown}
    \end{subfigure}
    \hfill
    \begin{subfigure}[t]{0.48\textwidth}
        \centering
        \includegraphics[width=\textwidth]{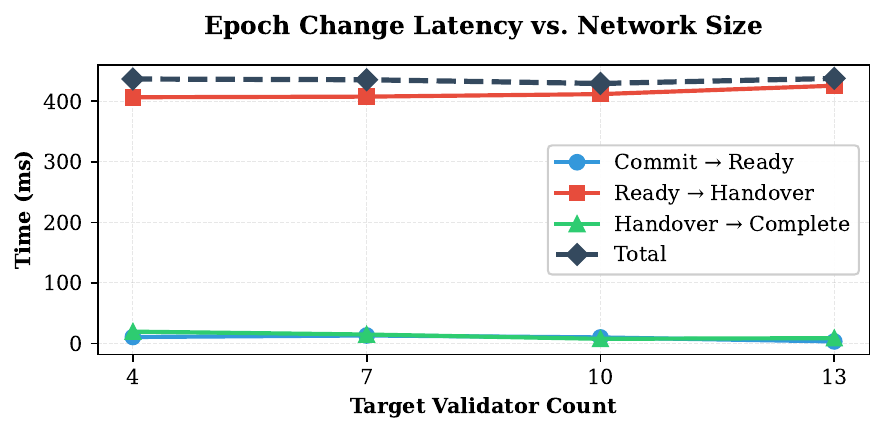}
        \caption{Latency scaling with validator count.}
        \label{fig:epoch-scaling}
    \end{subfigure}
    \label{fig:epoch-performance}
\end{figure*}

This modular architecture is being implemented as part of the Rialo blockchain~\cite{rialo2026blockchain}. 
It provides the starting point for the seamless integration of future upgrades to the consensus engine, 
execution engine or data dissemination layer. 

Initial results are promising. We evaluate the performance of our reconfiguration protocol on a local testbed to measure the latency of reconfiguration operations. We measure the time required to complete epoch transitions with varying validator set sizes. We test Rialo with 4 epoch transitions: from 4 to 4 validators (no size change), 4 to 7 validators, 7 to 10 validators, and 10 to 13 validators.

Figure~\ref{fig:epoch-breakdown} shows the performance characteristics of epoch transitions. We break down each transition into three phases: (1) from EpochChange commit to Ready message, (2) from Ready message to Handover (when the Ready quorum is reached), and (3) from Handover to completion.

The results show that the epoch change protocol performs efficiently across all tested configurations. We observe that the size of the validator set has a minimal impact on the overall latency of epoch transitions. The Ready to Handover phase takes approximately 93\% of the total epoch change latency. This shows the commit of Ready messages in consensus is the main bottleneck. We expect that this would further increase when we deploy the protocol in a geo-distributed setting.

\section{Related Work}

Reconfiguration in state machine replication has been carefully studied, both in the context of crash fault tolerance~\cite{lamport2009vertical,ongaro2014search,whittaker2021matchmaker,lorch2006smart,jehl2014arec,jehl2015smartmerge} and Byzantine fault tolerance~\cite{howard2023ccf,bessani2014smart,duan22dyno}.  Unfortunately, the majority of consensus algorithms today do not explicitly discuss reconfiguration. Those that do so propose algorithms that are specific to a particular protocol, and thus tightly coupled with it. 

\par \textbf{CFT Reconfiguration.} Classical state machine replication protocols like Paxos and Raft handle reconfiguration by treating configuration changes as special log entries, that must themselves go through consensus. The key challenge comes from handling the transition from one configuration to another. The work discusses a number of strategies to achieve that, often halting progress or limiting availability during the transition. \sys{} is able to achieve identical results without limiting pipelining at the cost of delaying reconfiguration taking effect for additional consensus rounds. 
In SMART~\cite{lorch2006smart}, reconfiguration of the system is managed
by creating an additional group of replicas. The two groups
of replicas run parallel Paxos instances until the system state is fully migrated to the new group. The popular\emph{Raft}~\cite{ongaro2014search} consensus protocol initially proposed a joint-consensus approach where majorities from both old and new configurations must overlap, later refining this to single-server membership changes to simplify safety. 

Vertical approaches to reconfiguration~\cite{lamport2009vertical} decouple the configuration from the consensus protocol. Vertical Paxos, for instance, allows the set of acceptors to change within a single instance by relying on an auxiliary \emph{configuration master} to manage membership. This ``vertical'' shift enables reconfiguration without stopping the stream of commands but introduces a dependency on an external, highly available master. Matchmaker Paxos~\cite{whittaker2021matchmaker} generalizes Vertical Paxos by replacing the auxiliary master with a set of ``matchmakers.'' These matchmakers persist configuration data, allowing the system to reconfigure without a stall in the pipeline or a single point of failure. 

\par \textbf{BFT Reconfiguration.} BFT-SMART~\cite{bessani2014smart} extends the ideas of Horizontal Paxos to the Byzantine setting, allowing replicas to be added or removed one at a time through a special reconfiguration command.  Duan et al.~\cite{duan22dyno} offers a formal treatment of BFT with dynamic membership - and highlights the different possible optimisations as a function of assumptions on the system, including assuming a fraction of correct replicas never leave the system.
BChain takes a different approach called rechaining for BFT protocols based on chain-replication. When a BFT service experiences failures or asynchrony, the head reorders the chain when a
replica is suspected to be faulty, so that a fault cannot affect the critical path.

\par \textbf{Permissionless Blockchains} Permissionless blockchain systems such as Ethereum require that the consensus protocol remain live even if 1) a large number of participants is offline 2) the system has high churn rates, where the set of participants changes frequently. Neu et al. \cite{neu2025limits} explores the theoretical limits of consensus under dynamic availability, where nodes can join or leave arbitrarily. In this orthogonal setting, it identifies necessary conditions for safety and highlights that existing protocols often make unrealistic assumptions such as requiring synchrony.

\section{Conclusion}
This paper introduced \sys{}, a novel reconfiguration protocol that supports both arbitrary membership changes and updates to consensus, while remaining fully modular. By distinguishing between the inner and outer logs, \sys{} allows for seamless transitions between different consensus implementations without tightly coupling the reconfiguration logic to any specific protocol.

\bibliographystyle{plain}
\bibliography{references}

@inproceedings{ongaro2014search,
  title={In search of an understandable consensus algorithm},
  author={Ongaro, Diego and Ousterhout, John},
  booktitle={Proceedings of the 2014 USENIX Annual Technical Conference (USENIX ATC 14)},
  pages={305--319},
  year={2014}
}

@article{lamport2010reconfiguring,
  title={Reconfiguring a state machine},
  author={Lamport, Leslie and Malkhi, Dahlia and Zhou, Lidong},
  journal={ACM SIGACT News},
  volume={41},
  number={1},
  pages={63--73},
  year={2010},
  publisher={ACM New York, NY, USA}
}

@inproceedings{yin2019hotstuff,
  author = {Yin, Maofan and Malkhi, Dahlia and Reiter, Michael K. and Gueta, Guy Golan and Abraham, Ittai},
  title = {HotStuff: BFT Consensus with Linearity and Responsiveness},
  booktitle = {Proceedings of the 2019 ACM Symposium on Principles of Distributed Computing (PODC)},
  year = {2019}
}

@inproceedings{lamport2009vertical,
  title={Vertical paxos and primary-backup replication},
  author={Lamport, Leslie and Malkhi, Dahlia and Zhou, Lidong},
  booktitle={Proceedings of the 28th ACM symposium on Principles of distributed computing},
  pages={312--313},
  year={2009}


}

@misc{sui2025mysticetiv2,
  author = {Sui},
  title = {Mysticeti: The Next Generation of Sui Consensus},
  year = {2025},
  howpublished = {\url{https://blog.sui.io/mysticeti-v2-sui-consensus/}},
  note = {Accessed: 2026-01-30}
}

@inproceedings{whittaker2021matchmaker,
  title={Matchmaker Paxos: A reconfigurable consensus protocol},
  author={Whittaker, Michael and Giridharan, Neil and Hellerstein, Joseph M and Howard, Heidi and Stoica, Ion and Crooks, Natacha},
  booktitle={Proceedings of the 18th USENIX Symposium on Networked Systems Design and Implementation (NSDI 21)},
  pages={67--84},
  year={2021}
  }

@misc{sui2026blockchain,
  author = {Sui Foundation},
  title = {Sui: A Next-Generation Smart Contract Platform},
  year = {2026},
  howpublished = {\url{https://www.sui.io/}},
  note = {Accessed: 2026-01-30}
}

@misc{rialo2026blockchain,
  author = {Rialo},
  title = {Rialo: The High-Performance Modular Blockchain for Distributed Systems},
  year = {2026},
  howpublished = {\url{https://www.rialo.io/}},
  note = {Accessed: 2026-01-30}
}

@article{schneider1990smr,
  author = {Schneider, Fred B.},
  title = {Implementing Fault-Tolerant Services Using the State Machine Approach: A Tutorial},
  journal = {ACM Computing Surveys},
  volume = {22},
  number = {4},
  pages = {299--319},
  year = {1990}
}

@misc{apache2023accord,
  author = {Apache Software Foundation},
  title = {CEP-15: Fast General Purpose Transactions (Accord)},
  year = {2023},
  howpublished = {\url{https://cwiki.apache.org/confluence/display/CASSANDRA/CEP-15%3A+Fast+General+Purpose+Transactions}},
  note = {Accessed: 2026-01-30}
}

@article{neu2025limits,
  title={On the Limits of Consensus under Dynamic Availability and Reconfiguration},
  author={Neu, Joachim and Nieto, Javier and Ren, Ling},
  journal={arXiv preprint arXiv:2510.03625},
  year={2025}
}

@inproceedings{castro1999pbft,
  author    = {Castro, Miguel and Liskov, Barbara},
  title     = {Practical Byzantine Fault Tolerance},
  booktitle = {Proceedings of the 3rd Symposium on Operating Systems Design and Implementation (OSDI)},
  year      = {1999},
  pages     = {173--186},
  publisher = {USENIX Association}
}

@inproceedings{babel2025mysticeti,
  author    = {Babel, Kushal and Chursin, Andrey and Danezis, George and Kichidis, Anastasios and Kokoris-Kogias, Lefteris and Koshy, Arun and Sonnino, Alberto and Tian, Mingwei},
  title     = {Mysticeti: Reaching the Latency Limits with Uncertified DAGs},
  booktitle = {Network and Distributed System Security (NDSS) Symposium},
  year      = {2025},
  publisher = {The Internet Society}
}

@article{antoniadis2023accord,
  author    = {Antoniadis, Karolos and Benhaim, Julien and Desjardins, Antoine and Poroma, Elias and Gramoli, Vincent and Guerraoui, Rachio and Voron, Gauthier and Zablotchi, Igor},
  title     = {Leaderless Consensus},
  journal   = {Journal of Parallel and Distributed Computing},
  year      = {2023},
  volume    = {176},
  pages     = {1--19}
}

@article{howard2023ccf,
author = {Howard, Heidi and Alder, Fritz and Ashton, Edward and Chamayou, Amaury and Clebsch, Sylvan and Costa, Manuel and Delignat-Lavaud, Antoine and Fournet, C\'{e}dric and Jeffery, Andrew and Kerner, Matthew and Kounelis, Fotios and Kuppe, Markus A. and Maffre, Julien and Russinovich, Mark and Wintersteiger, Christoph M.},
title = {Confidential Consortium Framework: Secure Multiparty Applications with Confidentiality, Integrity, and High Availability},
year = {2023},
issue_date = {October 2023},
publisher = {VLDB Endowment},
volume = {17},
number = {2},
issn = {2150-8097},
url = {https://doi.org/10.14778/3626292.3626304},
doi = {10.14778/3626292.3626304},
journal = {Proc. VLDB Endow.},
month = oct,
pages = {225–240},
numpages = {16}
}

@article{lorch2006smart,
author = {Lorch, Jacob R. and Adya, Atul and Bolosky, William J. and Chaiken, Ronnie and Douceur, John R. and Howell, Jon},
title = {The SMART way to migrate replicated stateful services},
year = {2006},
issue_date = {October 2006},
publisher = {Association for Computing Machinery},
address = {New York, NY, USA},
volume = {40},
number = {4},
issn = {0163-5980},
url = {https://doi.org/10.1145/1218063.1217946},
doi = {10.1145/1218063.1217946},
journal = {SIGOPS Oper. Syst. Rev.},
month = apr,
pages = {103–115},
numpages = {13}
}

@inproceedings{bessani2014smart,
  author    = {Bessani, Alysson and Sousa, João and Alchieri, Eduardo E. P.},
  title     = {State Machine Replication for the Masses with BFT-SMaRt},
  booktitle = {Proceedings of the 44th Annual IEEE/IFIP International Conference on Dependable Systems and Networks (DSN)},
  year      = {2014},
  pages     = {355--366}
}

@inproceedings{stathakopoulou2019mirbft,
  author    = {Stathakopoulou, Chrysoula and Tudor, Tudor David and Vukoli{\'c}, Marko},
  title     = {Mir-BFT: High-Throughput BFT for Blockchains},
  booktitle = {Proceedings of the 2nd Workshop on Scalable and Resilient Infrastructures for Distributed Ledgers (SRIDL)},
  year      = {2019}
}

@inproceedings{cowling2006hq,
  author = {Cowling, James and Myers, Daniel and Liskov, Barbara and Rodrigues, Rodrigo and Shrira, Liuba},
  title = {HQ Replication: A Hybrid Quorum Protocol for Byzantine Fault Tolerance},
  booktitle = {Proceedings of the 7th USENIX Symposium on Operating Systems Design and Implementation (OSDI)},
  year = {2006}
}

@inproceedings{kotla2007zyzzyva,
  author = {Kotla, Ramakrishna and Dahlin, Lorenzo and Clement, Allen and Wong, Edmund and Dahlin, Mike},
  title = {Zyzzyva: Speculative Byzantine Fault Tolerance},
  booktitle = {Proceedings of the 21st ACM Symposium on Operating Systems Principles (SOSP)},
  year = {2007}
}

@inproceedings{clement2009aardvark,
  author = {Clement, Allen and Wong, Edmund and Alvisi, Lorenzo and Dahlin, Mike and Marchetti, Mirco},
  title = {Making Byzantine Fault Tolerant Systems Tolerate Byzantine Faults},
  booktitle = {Proceedings of the 6th USENIX Symposium on Networked Systems Design and Implementation (NSDI)},
  year = {2009}
}

@inproceedings{miller2016honeybadgerbft,
  author = {Miller, Andrew and Xia, Yu and Croman, Kyle and Shi, Elaine and Song, Dawn},
  title = {The Honey Badger of BFT Protocols},
  booktitle = {Proceedings of the 2016 ACM SIGSAC Conference on Computer and Communications Security (CCS)},
  year = {2016}
}

@inproceedings{gueta2019sbft,
  author = {Gueta, Guy G. and Abraham, Ittai and Grossman, Shelly and Malkhi, Dahlia and Pinkas, Benny and Reiter, Michael and Seredinschi, Dragos and Tamir, Orr and Tomescu, Alin},
  title = {SBFT: A Scalable and Decentralized Trust Infrastructure},
  booktitle = {Proceedings of the 40th IEEE Symposium on Security and Privacy (Oakland)},
  year = {2019}
}

@inproceedings{neiheiser2021kauri,
  author = {Neiheiser, Ray and Matos, Miguel and Rodrigues, Luís},
  title = {Kauri: Scalable BFT Consensus with Pipelined Tree-Based Dissemination and Aggregation},
  booktitle = {Proceedings of the 28th ACM Symposium on Operating Systems Principles (SOSP)},
  year = {2021}
}

@inproceedings{spiegelman2022bullshark,
  author = {Spiegelman, Alexander and Giridharan, Neil and Sonnino, Alberto and Kokoris-Kogias, Lefteris},
  title = {Bullshark: DAG BFT Protocols Made Practical},
  booktitle = {Proceedings of the 2022 ACM SIGSAC Conference on Computer and Communications Security (CCS)},
  year = {2022}
}

@inproceedings{danezis2022tusk,
  author = {Danezis, George and Kokoris-Kogias, Lefteris and Sonnino, Alberto and Spiegelman, Alexander},
  title = {Narwhal and Tusk: A DAG-based Mempool and Efficient BFT Consensus},
  booktitle = {Proceedings of the 2022 European Conference on Computer Systems (EuroSys)},
  year = {2022}
}

@inproceedings{antunes2024aleabft,
  author = {Antunes, Diogo S. and Oliveira, Afonso N. and Breda, André and Franco, Matheus G. and Moniz, Henrique and Rodrigues, Rodrigo},
  title = {Alea-BFT: Practical Asynchronous Byzantine Fault Tolerance},
  booktitle = {Proceedings of the 21st USENIX Symposium on Networked Systems Design and Implementation (NSDI)},
  year = {2024}
}

@inproceedings{giridharan2024autobahn,
  author = {Giridharan, Neil and Suri-Payer, Florian and Abraham, Ittai and Alvisi, Lorenzo and Crooks, Natacha},
  title = {Autobahn: Seamless High-Speed BFT},
  booktitle = {Proceedings of the 30th ACM Symposium on Operating Systems Principles (SOSP)},
  year = {2024}
}

@INPROCEEDINGS{duan22dyno,
  author={Duan, Sisi and Zhang, Haibin},
  booktitle={2022 IEEE Symposium on Security and Privacy (SP)}, 
  title={Foundations of Dynamic BFT}, 
  year={2022},
  volume={},
  number={},
  pages={1317-1334},
  doi={10.1109/SP46214.2022.9833787}}

@inproceedings{jehl2014arec,
  author = {Jehl, Leander and Meling, Hein},
  title = {Asynchronous Reconfiguration for Paxos State Machines},
  booktitle = {Proceedings of the 15th International Conference on Distributed Computing and Networking (ICDCN)},
  year = {2014},
  url = {https://www.ux.uis.no/~meling/papers/2014-asyncreconfig-icdcn.pdf}
}

@inproceedings{jehl2015smartmerge,
  author = {Jehl, Leander and Vitenberg, Roman and Meling, Hein},
  title = {SmartMerge: A New Approach to Reconfiguration for Atomic Storage},
  booktitle = {Proceedings of the 29th International Symposium on Distributed Computing (DISC)},
  year = {2015},
  url = {https://doi.org/10.1007/978-3-662-48653-5_11}
}

\pagebreak

\end{document}